\documentclass[12pt]{article}

\usepackage{graphicx,xcolor}
\usepackage{amssymb,amsfonts,amsthm}
\usepackage{cancel}
\usepackage{hyperref}
\newcommand{\bO}{{O}}
\newcommand{\bo}{{o}}

\newcommand{\EE}{\mathbb{E}}
\newcommand{\VV}{\mathbb{V}}
\newcommand{\PP}{\mathbb{P}}
\newcommand{\RR}{\mathbb{R}}

\newcommand{\fracs}[2]{{ \textstyle \frac{#1}{#2} }}

\def \one  {{\mathbf{1}}}
\def \D    {{\rm d}}

\def \eps  {{\varepsilon}}

\def \hD {{\widehat{D}}}
\def \hP {{\widehat{P}}}

\def \hS {{\widehat{S}}}

\def \hY {{\widehat{Y}}}

\def \hp {{\widehat{p}}}

\def \bS {{\overline{S}}}
\def \hbS {{\overline{\hS}}}

\newtheorem{theorem}{Theorem}[section]
\newtheorem{lemma}[theorem]{Lemma}

\newtheorem{corollary}[theorem]{Corollary}

\newcommand{\added}[1]{{\color{red}#1}}
\newcommand{\replaced}[2]{\bcancel{#1}{\color{red}#2}}

\begin{document}

\title{Analysis of multilevel Monte Carlo path simulation using the Milstein discretisation}

\author{
Michael B.~Giles\\
\texttt{mike.giles@maths.ox.ac.uk}
\and
Kristian Debrabant\\
\texttt{debrabant@imada.sdu.dk}
\and
Andreas R\"{o}{\ss}ler\\
\texttt{roessler@math.uni-luebeck.de}
}

\maketitle

\begin{abstract}
The multilevel Monte Carlo path simulation method introduced by Giles 
({\it Operations Research}, 56(3):607-617, 2008) exploits strong convergence 
properties to improve the computational complexity by combining simulations 
with different levels of resolution.  In this paper we analyse its 
efficiency when using the Milstein discretisation; this has an improved 
order of strong convergence compared to the standard Euler-Maruyama method,
and it is proved that this leads to an improved order of convergence of the 
variance of the multilevel estimator.  Numerical results are also given for 
basket options to illustrate the relevance of the analysis.
\end{abstract}



\section{Introduction}

In computational finance, Monte Carlo methods are used to estimate $\EE[P]$,
the expected value of a discounted option payoff function which depends 
on the solution of an SDE of the generic form 
\begin{equation}
\D S(t) = a(S(t),t) \, \D t + b(S(t),t)\, \D W(t), \quad 0 \leq t \leq T,
\label{eq:SDE}
\end{equation}
subject to specified initial data $S(0)\!=\!S_0$.

Using a standard Monte Carlo method with a numerical discretisation with 
first order weak convergence, to achieve an RMS error of $\eps$ would 
require $\bO(\eps^{-2})$ independent paths, each with $\bO(\eps^{-1})$ 
timesteps, giving a computational complexity which is $\bO(\eps^{-3})$.  
However, the multilevel Monte Carlo (MLMC) approach of Giles 
\cite{giles08,giles08b} reduces the cost to $\bO(\eps^{-2})$ under 
certain circumstances.

The key identity underlying the method is
\begin{equation}
\EE[\hP_L] = \EE[\hP_0] + \sum_{\ell=1}^L \EE[\hP_\ell \!-\! \hP_{\ell-1}].
\label{eq:identity}
\end{equation}
This expresses the expectation on the finest level of resolution, 
using $2^{\bcancel{-}L}$ uniform timesteps, as the sum of the expected value 
on level 0, using just one timestep of size $T$, plus a sum of 
expected corrections between levels $\ell$ and $\ell\!-\!1$.
The quantity $\EE[\hP_\ell \!-\! \hP_{\ell-1}]$ can be estimated using
$N_\ell$ independent samples by
\begin{equation}
\hY_\ell = N_\ell^{-1} \sum_{i=1}^{N_\ell} \left( \hP_\ell^{(i)} \!-\! \hP_{\ell-1}^{(i)} \right).
\label{eq:est_l}
\end{equation}
Note that the difference $\hP_\ell^{(i)} \!-\! \hP_{\ell-1}^{(i)}$ comes from two 
discrete approximations with different timesteps but the same Brownian path; 
this difference is often small because of the strong convergence properties of 
the numerical discretisation scheme. The variance of this simple estimator is 
$\displaystyle
\VV[\hY_\ell] = N_\ell^{-1} V_\ell
$
where $V_\ell$ is the variance of a single sample.  It is the convergence of
$V_\ell$ as $\ell\rightarrow \infty$ which is the focus of this paper, because 
if $V_\ell = O(2^{-\beta \ell})$ with $\beta>1$ then an overall RMS accuracy
of $\eps$ can be achieved at a computational cost which is $O(\eps^{-2})$
\cite{giles08,giles15} through selecting an optimal number of samples 
on each level.

\begin{table}
\begin{center}
\begin{tabular}{|l|l|l|l|l|}
\hline & \multicolumn{2}{c|}{Euler-Maruyama} &  \multicolumn{2}{c|}{Milstein} \\ 
option & numerical  & analysis & numerical & analysis  \\ \hline
Lipschitz & $\bO(h)$       & $\bO(h)$              & $\bO(h^2)$ & $\bO(h^2)$ \\
Asian     & $\bO(h)$       & $\bO(h)$              & $\bO(h^2)$ & $\bO(h^2)$ \\
lookback  & $\bO(h)$       & $\bO(h)$              & $\bO(h^2)$ & $\bO(h^2 (\log h)^2)$ \\
barrier   & $\bO(h^{1/2})$ & $\bo(h^{1/2-\delta})$ & $\bO(h^{3/2})$ & $\bo(h^{3/2-\delta})$ \\
digital   & $\bO(h^{1/2})$ & $\bO(h^{1/2}\added{|}\log h\added{|})$ & $\bO(h^{3/2})$ & $\bo(h^{3/2-\delta})$ \\ \hline
\end{tabular}
\end{center}

\caption{Orders of convergence for $V_\ell$ as observed numerically and 
proved analytically for both the Euler-Maruyama and Milstein discretisations; 
$\delta$ can be any strictly positive constant.}

\end{table}

For the MLMC method based on the simple Euler-Maruyama discretisation
with a uniform timestep of size $h$, 
Giles, Higham and Mao \cite{ghm09} proved that $V_\ell = \bO(h)$ for 
European options (based on the final value of the underlying $S(T)$)
with a uniform Lipschitz payoff, Asian options (based on the average 
value of the underlying) and lookback options (based on the minimum 
or maximum of the underlying). 
They also proved that $V_\ell = \bo(h^{1/2-\delta})$, for any $\delta>0$, 
for barrier options (in which the payoff is zero if the underlying 
asset crosses, or fails to cross, a certain level) and digital options 
(for which the payoff is a discontinuous function of $S(T)$).  
The final result has been tightened by Avikainen \cite{avikainen09} 
who proved in that case that $V_\ell = \bO(h^{1/2}\added{|}\log h\added{|})$.  As summarised 
in Table 1, numerical simulations \cite{giles08} suggest that all of 
these results are near-optimal.

For the MLMC method based on the Milstein discretisation,
numerical simulations \cite{giles08b} suggest that $V_\ell=\bO(h^2)$ for 
European options with a uniform Lipschitz payoff and for Asian and 
lookback options, and $V_\ell=\bO(h^{3/2})$ for the barrier and digital 
options.  In this paper we aim to establish these orders of 
convergence analytically, and do so near optimally in each case.

The numerical analysis will be performed for scalar SDEs, but we begin 
the paper by presenting numerical results for basket options, in which 
each of the underlying assets has a drift 
and volatility which does not depend on the value of the other assets.
Under these conditions, the Milstein discretisation can be applied to
each asset individually, avoiding the need to simulate L{\'e}vy areas
as required in general for multi-dimensional SDEs.

\section{Numerical results for basket options}

\subsection{Milstein discretisation}

The Milstein discretisation of equation (\ref{eq:SDE}) for a
scalar SDE using a uniform timestep $h$ is
\begin{equation}
\hS_{n+1} = \hS_{n} + a_n\, h + b_n\, \Delta W_n
 + \fracs{1}{2} \, b'_n\, b_n\, 
\left(\rule{0in}{0.15in} (\Delta W_n)^2 - h \right)
\label{eq:Milstein1}
\end{equation}
In the above equation, the subscript $n$ is used to denote the timestep 
index with $t_n\equiv nh$ and $a_n, b_n, b'_n$ correspond to 
$a, b, \partial b/\partial S$ evaluated at $\hS_n,t_n$.
In addition, $\Delta W_n$ is the Brownian path increment $W_{n+1}\!-\!W_n$,
where $W_n \equiv W(t_n)$. 
Kloeden \& Platen \cite{kp92} proved that under certain conditions, which
will be defined later, the Milstein scheme
gives $O(h)$ strong convergence, and for a Lipschitz European payoff this 
immediately leads to the result that $V_\ell = O(h_\ell^2)$. 
This remains true for a put or call basket option based on the arithmetic 
average of several underlying assets, each of which is simulated using the 
Milstein discretisation.

Reference \cite{giles08b} developed and tested MLMC treatments for
the tougher challenges of Asian, lookback, barrier and digital options 
based on a single underlying asset; it is the numerical analysis of
those treatments which is the focus of this paper.
The key to the Asian, lookback and barrier option constructions is a 
conditional piecewise Brownian interpolation. Within the time interval 
$[t_n, t_{n+1}]$ we approximate the drift and volatility as being constant
and use a Brownian interpolation conditional on the two 
end values $\hS_n$ and $\hS_{n+1}$, giving
\begin{equation}
\hS(t) = \hS_n + \lambda(t)\, (\hS_{n+1}\!-\!\hS_n)
      + b_n \left(\rule{0in}{0.16in} W(t) - W_n - \lambda(t)\, (W_{n+1}\!-\!W_n) \right),
\label{eq:interp}
\end{equation}
where $\lambda(t) = (t \!-\! t_n)/h$.
Standard results for the distribution of the extrema and averages of 
Brownian motions \cite{glasserman04} can then be used to construct 
suitable multilevel estimators \cite{giles08b}; more details are given 
later in the numerical analysis section.

These treatments for a single underlying asset extend very naturally 
to basket options. If the value is dependent on a weighted average of 
$J$ underlying assets, $\bS(t) = \sum_{j=1}^J \mu_j\,S_j(t)$,
each of which satisfies an SDE of the form (\ref{eq:SDE}) driven by 
Brownian motions $W_j(t)$ with correlation matrix $\Sigma$, then
the important observation is that the weighted average of the Brownian 
interpolations for the $J$ underlying assets gives
\begin{eqnarray*}
\hbS(t) &\!=\!& \hbS_n + \lambda(t)\, (\hbS_{n+1}\!-\!\hbS_n) \\
&&\hspace{0.5in}
  +\ \sum_{j=1}^J \mu_j \, b_{j,n} \left(\rule{0in}{0.16in} W_j(t) - W_{j,n} - \lambda(t)\, (W_{j,n+1}\!-\!W_{j,n}) \right)
\\ &\!=\!& \hbS_n + \lambda(t)\, (\hbS_{n+1}\!-\!\hbS_n)
 +  \overline{b}_n \left(\rule{0in}{0.16in} W(t) - W_{n} - \lambda(t)\, (W_{n+1}\!-\!W_{n}) \right)
\end{eqnarray*}
where $W(t)$ is another scalar Brownian motion which is a weighted combination 
of the $W_j(t)$, and
$\overline{b}_n$ is defined by 
\[
\overline{b}^2_n = \sum_{i,j} \mu_i \, b_{i,n}\, \Sigma_{i,j}\, \mu_j\,b_{j,n}.
\]
Since the Brownian interpolation for the basket average has the same form 
as the scalar interpolation, the multilevel estimators can be constructed 
in exactly the same way as in \cite{giles08b}, but using $\overline{b}_n$.

Similarly, for the digital option which has a discontinuous payoff
based on a single underlying, one can use a constant coefficient Brownian 
extrapolation conditional on the value $\hS_{N-1}$, one timestep before 
the end.  Following an approach used for payoff smoothing for pathwise 
sensitivity analysis \cite{glasserman04}, the conditional expectation 
for the payoff can be evaluated analytically and this is then used to 
construct the multilevel estimator \cite{giles08b}.  This treatment 
also extends naturally to the basket case by considering the weighted
average of the $J$ conditional extrapolations.

\subsection{Numerical results}

The numerical results to be presented are for a basket of five assets, 
each modelled as a geometric Brownian motion:
\[
\D S_j(t) = r \, S_j(t) \ \D t + \sigma_j \, S_j(t) \ \D W_j(t), \quad 0<t<T,
\]
using a constant risk-free interest rate $r\!=\!0.05$, and five
volatilities $\sigma = 0.2, 0.25, 0.3, 0.35, 0.4$.  The initial asset 
values are $S_j(0) = 100$, the simulation interval is taken to be 
$T\!=\!1$, and the driving Brownian motions have a correlation of 0.25.
In each case, the option price is based on a simple arithmetic average
of the five assets.

These tests are based on ones performed previously in \cite{giles09},
but they have been updated to use the latest version of the high-level
MLMC driver software \cite{giles15} which estimates the rates of 
convergence of the weak error and the variance $V_\ell$, and therefore
determines (approximately) the optimal number of levels $L$ and the 
number of samples $N_\ell$ on each level.
\footnote{These numerical tests are included in the MATLAB software available online at
{\tt http://people.maths.ox.ac.uk/gilesm/mlmc/}
and an archived version is kept on
{\tt http://people.maths.ox.ac.uk/gilesm/mlmc.html}
}

\subsubsection{Asian Option}

\begin{figure}[t!]
\begin{center}
\includegraphics[width=4.2in]{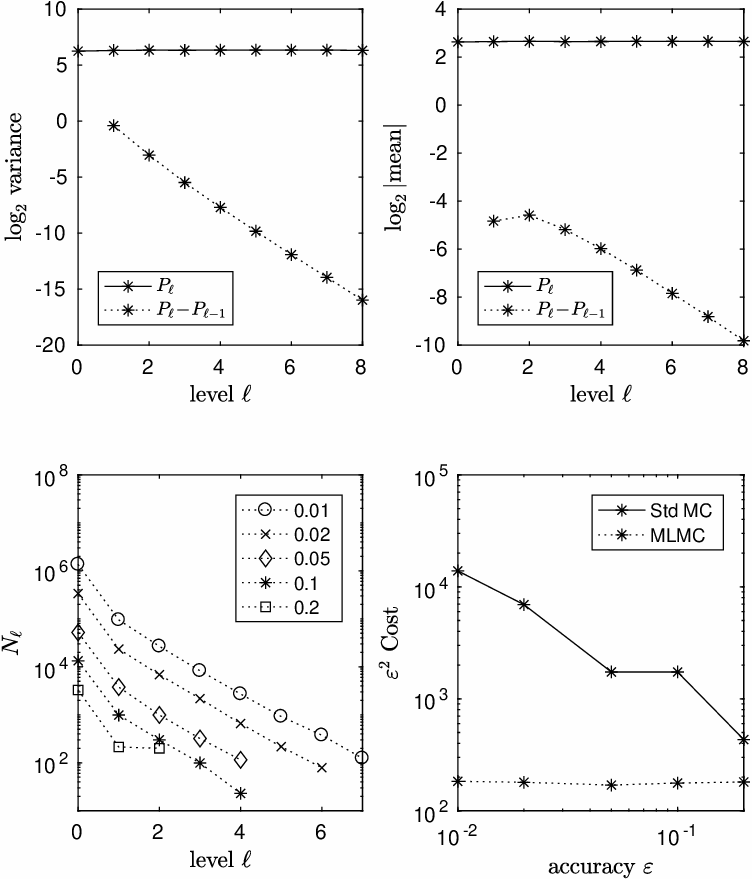}
\end{center}
\vspace{-0.2in}
\caption{Asian option}
\label{fig:GBM2}
\end{figure}

The Asian basket option has discounted payoff
$\displaystyle
P = \exp(-rT)\ \max\left(0, \overline{\overline{S}} \!-\! K  \right),
$
where $\overline{\overline{S}}$ is the time-average of the average 
of the underlying assets, and the strike is $K\!=\!100$.

The top left plot in Figure \ref{fig:GBM2} shows the behaviour of the 
variance of both $\hP_\ell$ and $\hP_\ell-\hP_{\ell-1}$.  The slope of the 
latter is approaching a value approximately equal to $-2$, indicating 
that $V_\ell\!=\!O(h_\ell^2)$. On level $\ell\!=\!2$, 
which has just 4 timesteps, $V_\ell$ is already almost 1000 times smaller 
than the variance $\VV[\hP_\ell]$ of the standard Monte Carlo method with the 
same timestep.
The top right plot shows that $|\EE[\hP_\ell\!-\!\hP_{\ell-1}]|$ is approximately 
$O(h_\ell)$, corresponding to first order weak convergence.
This is used to determine the number of levels that are required to reduce 
the bias to an acceptable level \cite{giles08b}.

The bottom two plots have results from five multilevel calculations for
different values of $\eps$.  Each line in the bottom left plot shows the values 
for $N_\ell, \ell=0,\ldots, L$, with the values decreasing with $\ell$ because of the 
decrease in both $V_\ell$ and $h_\ell$. It can also be seen that the value for $L$, 
the maximum level of timestep refinement, increases as the value for $\eps$ 
decreases, requiring a lower bias error \cite{giles08b}.
The bottom right plot shows the variation with $\eps$ of $\eps^2\,C$ 
where the computational complexity $C$ is defined as
$\displaystyle
C = \sum_\ell 2^\ell N_\ell,
$
which is the total number of fine grid timesteps on all levels.
One line shows the results for the multilevel
calculation and the other shows the corresponding cost of a standard Monte Carlo 
simulation of the same accuracy, i.e.~the same bias error corresponding to the 
same value for $L$, and the same variance.  It can be seen that $\eps^2 C$ is 
almost constant for the multilevel method, as expected, whereas for the standard 
Monte Carlo method it increases with $L$.  For the most accurate case, 
$\eps\!=\!0.01$, the multilevel method is approximately 100 times more efficient 
than the standard method.

\subsubsection{Lookback Option}

\begin{figure}[t!]
\begin{center}
\includegraphics[width=4.2in]{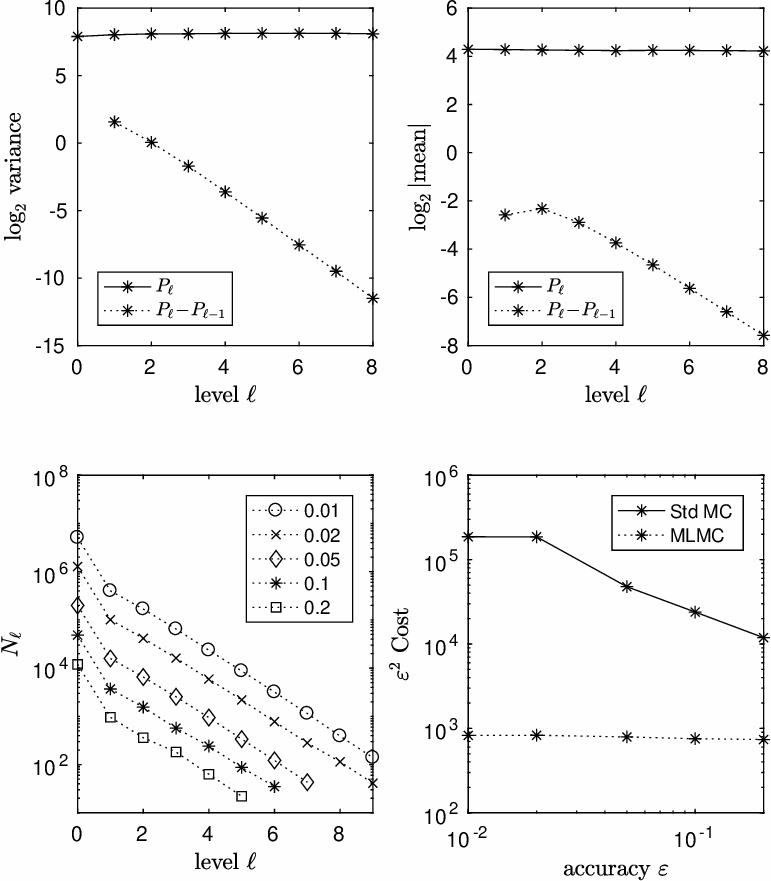}
\end{center}
\vspace{-0.2in}
\caption{Lookback option}
\label{fig:GBM3}
\end{figure}

The basket lookback option we consider has the discounted payoff
$\displaystyle
P = \exp(-rT) \left( \overline{S}(T) - \min_{0<t<T} \overline{S}(t) \right).
$

The top left plot in Figure \ref{fig:GBM3} shows that the variance is $O(h_\ell^2)$, 
while the top right plot shows that the mean correction is $O(h_\ell)$.  The bottom 
left plot shows that more levels are required to reduce the discretisation bias to the 
required level.  Consequently, the savings relative to the standard Monte Carlo treatment 
are greater, up to a factor of approximately 200 for $\eps\!=\!0.01$.  The computational 
cost of the multilevel method is almost perfectly proportional to $\eps^{-2}$.

\subsubsection{Barrier Option}

\begin{figure}[t!]
\begin{center}
\includegraphics[width=4.2in]{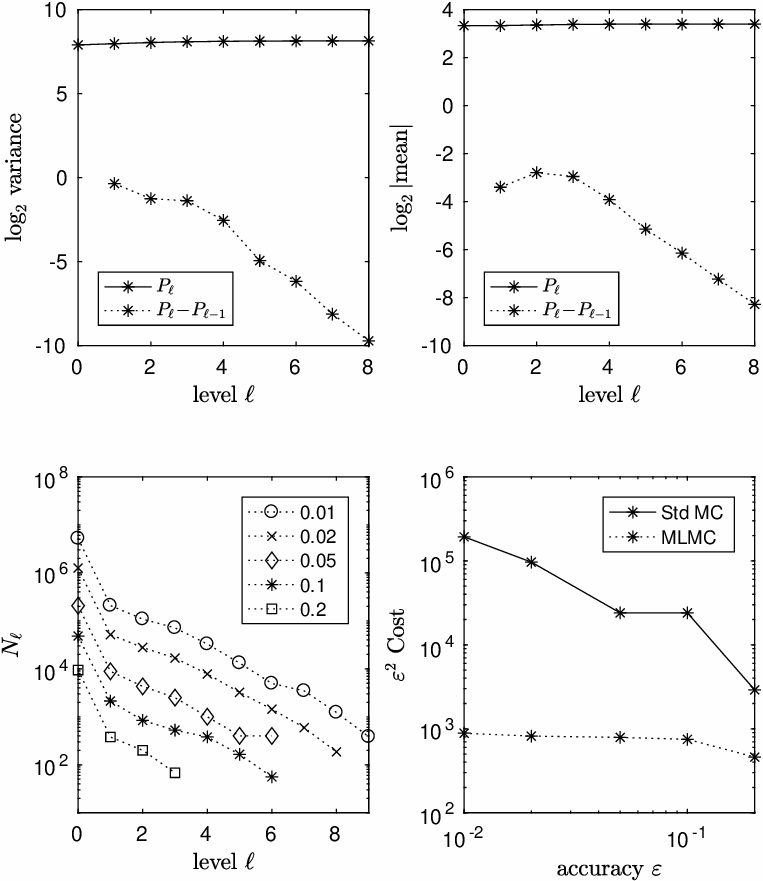}
\end{center}
\vspace{-0.2in}
\caption{Barrier option}
\label{fig:GBM4}
\end{figure}

The barrier option which is considered is a down-and-out call 
with payoff
$\displaystyle
P = \exp(-rT)\ (\overline{S}(T) - K)^+ \ {\bf 1}_{\tau>T},
$
where the notation $(x)^+$ denotes $\max(0,x)$, 
${\bf 1}_{\tau>T}$ is an indicator function taking value 1 if
the argument is true, and zero otherwise, and the crossing time $\tau$ is defined
as 
$\displaystyle
\tau = \inf_{t>0}\left\{ \overline{S}(t)<B \right\}.
$
The barrier value is taken to be $B\!=\!85$, and the strike is again $K\!=\!100$.

The top left plot in Figure \ref{fig:GBM4} shows that the variance is 
approximately $O(h_\ell^{3/2})$.  The reason for this is that an $O(h_\ell^{1/2})$ 
fraction of the paths have a minimum which lies within $O(h_\ell^{1/2})$ of the 
barrier.  It will be proved that for these paths the difference 
between the coarse and fine path payoff values is $O(h_\ell^{1/2})$, giving a 
contribution to the overall variance which is $O(h_\ell^{3/2})$.

The top right plot shows that the mean correction is $O(h_\ell)$, corresponding to 
first order weak convergence.  The bottom right plot shows that the computational 
cost of the multilevel method is again almost perfectly proportional to $\eps^{-2}$, 
and for $\eps\!=\!0.01$ it is 200 times more efficient than the standard Monte Carlo 
method.

\subsubsection{Digital Option}

\begin{figure}[t!]
\begin{center}
\includegraphics[width=4.2in]{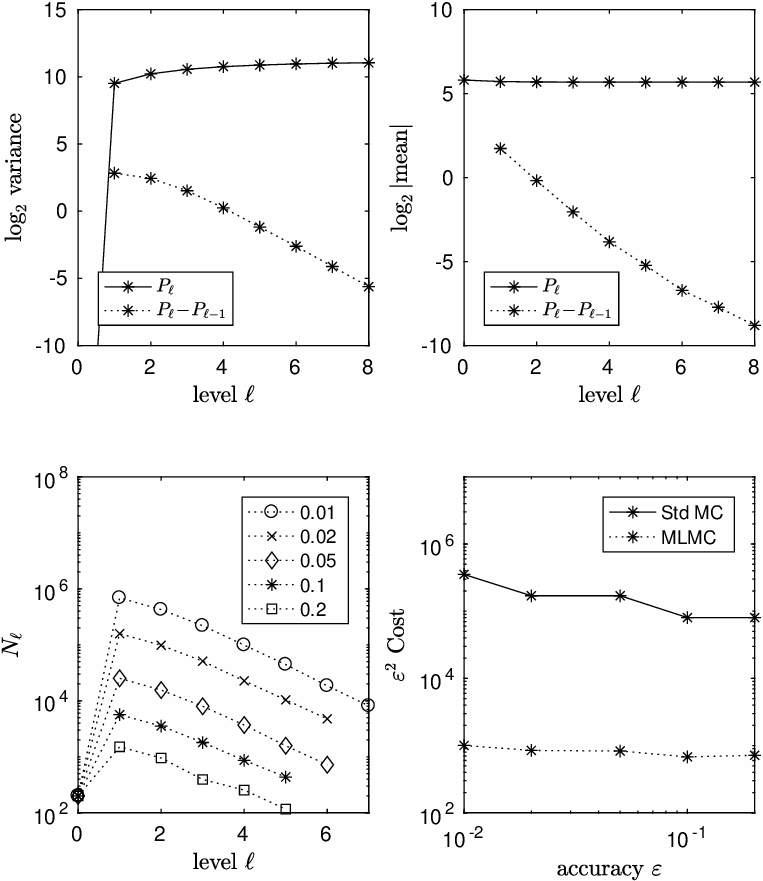}
\end{center}
\vspace{-0.2in}
\caption{Digital option}
\label{fig:GBM5}
\end{figure}

The digital option has the discounted payoff
$\displaystyle
P = \exp(-rT)\ K \ {\bf 1}_{S(T)>K}
$
with strike $K\!=\!100$.

The top left plot in Figure \ref{fig:GBM5} shows that the variance is approximately 
$O(h_\ell^{3/2})$.  The reason for this is similar to the argument for the barrier option.  
$O(h_\ell^{1/2})$ of the paths have a \replaced{minimum}{terminal value} which lies within $O(h_\ell^{1/2})$ of the 
strike.  The fine path and coarse path trajectories differ by $O(h_\ell)$, due to the 
first order strong convergence of the Milstein scheme and it will be proved that this 
results in an $O(h_\ell^{1/2})$ difference between the coarse and fine path evaluations.

One strikingly different feature is that the variance of the level 0 estimator
is zero.  This is because the multilevel treatment introduced in \cite{giles08b} uses
a conditional expectation (based on a simple Brownian extrapolation for which the
expectation is known analytically) evaluated one timestep before the end.  At level 
$l\!=\!0$ where there would usually be one timestep, there is no path simulation at all; 
one simply uses the analytic expression for the conditional expectation.  This reduces 
the cost of the multilevel calculations even more than usual, giving more than a factor 
of 400 computational savings for $\eps\!=\!0.01$.

\section{Numerical analysis}

For clarity, the numerical analysis is now presented for financial options based on a 
single underlying asset, but for the reasons already discussed the numerical analysis
extends immediately to basket options.

\subsection{Preliminaries}

The analysis builds on a large body of existing results 
which are included here for completeness.

\subsubsection{Solution of the SDE and its Milstein discretisation}

We will assume throughout that the SDE (\ref{eq:SDE}) is scalar, and 
the drift $a \in C^{2,1}(\RR \!\times\! \RR^+)$ and 
volatility $b \in C^{3,1}(\RR \!\times\! \RR^+)$
satisfy the following standard conditions in which we use the notation
$L_0 \equiv \partial/\partial t + a\, \partial/\partial S\added{+\frac12b^2\, \partial^2/\partial S^2}$ and
$L_1 \equiv b\, \partial/\partial S$.

\begin{itemize}
\item
A1 (uniform Lipschitz condition): there exists $K_1$ such that 
\[
\left| a(x,t) - a(y,t) \right| +
\left| b(x,t) - b(y,t) \right| +
\left| L_1 b(x,t) - L_1 b(y,t) \right|\ \leq\ K_1\, | x \!-\! y |
\]
\item
A2 (linear growth bound): there exists $K_2$ such that 
\begin{eqnarray*}\ 
\ \left| a(x,t) \right| + \left| L_0 a(x,t) \right| + \left| L_1 a(x,t) \right|
+ \left| b(x,t) \right| + \left| L_0 b(x,t) \right| \\
+ \left| L_1 b(x,t) \right| + \left| L_0 L_1 b(x,t) \right| + \left| L_1 L_1 b(x,t) \right|
&\leq& K_2\, (1 + |x|) \nonumber
\end{eqnarray*}
\item
A3 (additional Lipschitz condition): there exists $K_3$ such that 
\[
\left| b(x,t) - b(x,s) \right|\ \leq\ K_3\, (1 + |x|)\, \sqrt{| t \!-\! s |}
\]
\end{itemize}

Under these conditions, we have the following result for the analytic solution
to the SDE \cite{kp92}.

\begin{theorem}
\label{thm:KP0}
Provided the assumptions A1-A3 are satisfied, then for all positive 
integers $m$,
$\displaystyle
\EE\left[\sup_{0<t<T} |S(t)|^m\right] < \infty.
$
\end{theorem}

\vspace{0.1in}

Kloeden \& Platen \cite{kp92} define the following continuous time 
interpolant of the Milstein discretisation in (\ref{eq:Milstein1})
for $t_n \!\leq\! t \!\leq\! t_{n+1}$,
\begin{eqnarray}
\hS_{KP}(t) &=&  \hS_{n} + a_n\, (t\!-\!t_n) + b_n\, (W(t)\!-\!W_n) 
\nonumber \\ && \hspace{0.22in}
 +\ \fracs{1}{2} \, b'_n\, b_n\, \left(  (W(t) \!-\! W_n)^2 - (t\!-\!t_n) \right),
\label{eq:KP}
\end{eqnarray}
and prove the following result.

\begin{theorem}
\label{thm:KP}
Provided the assumptions A1-A3 are satisfied, then for all positive integers $m$ 
there exists a constant $C_m$ such that
\[
\EE\left[\sup_{0<t<T} |S(t) - \hS_{KP}(t)|^m\right] < C_m\, h^m, \quad \quad
\EE\left[\sup_{0<t<T} |\hS_{KP}(t)|^m\right] < C_m.
\]
\end{theorem}

\subsubsection{Brownian bridge results}

If the drift $a$ and volatility $b$ are constant, the scalar SDE (\ref{eq:SDE})
has solution 
\[
S(t) = S_0 + a\, t + b\, W(t)
\]
and hence within the time interval $[t_n, t_{n+1}]$ of length $h$ we have
\begin{equation}
S(t) = S_n + \lambda(t)\,  (S_{n+1}\!-\!S_n) + b\, \left( W(t) - W_n - \lambda(t)\, (W_{n+1}\!-\!W_n) \right)
\label{eq:bridge}
\end{equation}
where again $\lambda(t) \equiv (t - t_n)/h$. 
This means that the deviation of $S(t)$ from a piecewise linear interpolation 
of the values $S_n \equiv S(t_n)$ is proportional to the deviation of $W(t)$ 
from its piecewise linear interpolation.  It can be proved that the distribution 
of the latter is independent of the Brownian increment $W_{n+1}\!-\!W_n$, and 
furthermore we have the following results (see for example \cite{glasserman04}).

\begin{lemma}
\label{thm:asian1}
Conditional on $S_n$ and $S_{n+1}$, the distribution for the integral of 
$S(t)$ over the interval $[t_n, t_{n+1}]$ is given by 
\begin{equation}
\int_{t_n}^{t_{n+1}} S(t) \, \D t = \fracs{1}{2} h \left( S_n + S_{n+1} \right)
+ b\, I_n
\end{equation}
where 
\[
I_n \equiv \int_{t_n}^{t_{n+1}}  \left( W(t) - W_n - \lambda(t)\, (W_{n+1}\!-\!W_n) \right)\, \D t
\]
is a $N(0,\fracs{1}{12}h^3)$ Normal random variable, independent of $W_{n+1}\!-\!W_n$.
\end{lemma}

\begin{lemma}
\label{thm:lookback1}
Conditional on $S_n$ and $S_{n+1}$, the distributions for the minimum and 
maximum of $S(t)$ over the interval $[t_n, t_{n+1}]$ are given by 
\begin{eqnarray*}
S_{n,min} &=& \fracs{1}{2} \left( S_n + S_{n+1} 
- \sqrt{ \left(S_{n+1} \!-\! S_n\right)^2 - 2\, b^2 \, h \log U_n }\ \right), \\
S_{n,max} &=& \fracs{1}{2} \left( S_n + S_{n+1} 
+ \sqrt{ \left(S_{n+1} \!-\! S_n\right)^2 - 2\, b^2 \, h \log V_n }\ \right),
\end{eqnarray*}
where $U_n$ and $V_n$ are each uniformly distributed on $(0,1)$.
\end{lemma}

\begin{lemma}
\label{thm:barrier1}
Provided $b\neq 0$, conditional on $S_n$ and $S_{n+1}$, the probability that the 
minimum (or maximum) of $S(t)$ over the interval $[t_n, t_{n+1}]$ is less than 
(or greater than) some value $B$, is
\begin{eqnarray*}
\hspace*{-0.2in}
\PP\left( \inf_{[t_n, t_{n+1}]} S(t) < B \ | \ S_n, S_{n+1}\right) 
&\!=\!& \exp\left(\frac{ - 2\, (S_n \!-\! B)^+(S_{n+1} \!-\! B)^+ }{b^2\, h} \right),
\hspace{0.2in} \\
\PP\left( \sup_{[t_n, t_{n+1}]} S(t) > B \ | \ S_n, S_{n+1}\right) 
&\!=\!& \exp\left(\frac{ - 2\, (B \!-\! S_n)^+(B \!-\! S_{n+1})^+ }{b^2\, h} \right).
\end{eqnarray*}
\end{lemma}

\begin{corollary}
\label{thm:barrier2}
If $W(t)$ is a Brownian motion with $W(0)\!=\!W(1)\!=\!0$, then for $x>0$
\[
\PP\left( \sup_{[0,1]} W(t) > x \right)\ =\ 
\PP\left( \inf_{[0,1]} W(t) < -x \right)\ =\  
\exp(- 2 x^2),
\]
and hence
$
\EE\left[ \sup_{[0,1]} |W(t)|^m \right]
$
is finite for all positive integers $m$.
\end{corollary}
 
\subsubsection{Extreme values}

The following results come from extreme value theory which determine
the limiting distribution of the maximum of a large set of i.i.d.~random
variables \cite{ekm08}.

\begin{lemma}
If $U_n, n=1, \ldots, N$ are independent samples from a uniform distribution 
on the unit interval $[0,1]$, then for any positive integer $m$
\begin{equation}
\EE\left[ \max_n |\log U_n|^m \right] = \bO((\log N)^m), \mbox{ as } N\rightarrow \infty.
\end{equation}
\label{thm:extreme1}
\end{lemma}

\begin{lemma}
If $Z_n, n=1, \ldots, N$ are independent samples from a standard Normal distribution,
then for any positive integer $m$
\begin{equation}
\EE\left[ \max_n |Z_n|^m \right] = \bO( (\log N)^{m/2} ), \mbox{ as } N\rightarrow \infty.
\end{equation}
\label{thm:extreme2}
\end{lemma}

\begin{corollary}
If $W_n(t), n=1, \ldots, N$ are independent Brownian paths on $[0,1]$, 
conditional on $W_n(0)=W_n(1)=0$, then for any positive integer $m$
\begin{equation}
\EE\left[ \max_n \sup_{[0,1]} |W_n(t)|^m \right] = \bO( (\log N)^{m/2} ), \mbox{ as } N\rightarrow \infty.
\end{equation}
\label{thm:extreme3}
\end{corollary}
\begin{proof}
From Corollary \ref{thm:barrier2}, for sufficiently large $x$ the tail
probability for $|W_n(t)|$ is less than that of a standard Normal random variable.
\end{proof}

\subsubsection{Extreme paths}

Some of the proofs in \cite{ghm09} use an argument that certain 
``extreme'' paths make a negligible contribution to the overall 
expectation.  This same argument will be employed in this paper
but in a more compact form based on these two lemmas.

\begin{lemma}
\label{thm:ep1}
If $X_\ell$ is a scalar random variable defined on level $\ell$ 
of the multilevel analysis, and for each positive integer $m$,
$\EE[ \, |X_\ell|^m]$ is uniformly bounded, then, for any $\delta>0$,
\[
\PP\left( \, |X_\ell| > h_\ell^{-\delta} \right) = \bo(h_\ell^p), \qquad \forall p>0.
\]
\end{lemma}
\begin{proof}
Follows immediately from Markov's inequality 
\[
\PP\left( |X_\ell| \geq h_\ell^{-\delta}  \right)\ = \
\PP\left( |X_\ell|^m \geq h_\ell^{-m \delta} \right)\ \leq\ h_\ell^{m \delta}\ \EE[\,|X_\ell|^m],
\]
by choosing $m > p/\delta$.
\end{proof}

\begin{lemma}
\label{thm:ep2}
If $Y_\ell$ is a scalar random variable on level $\ell$,
$\EE[\, Y_\ell^2 ]$ is uniformly bounded, 
and for each $p\!>\!0$, the indicator function $\one_{E_\ell}$
on level $\ell$ (which takes value 1 or 0 depending whether or not
a path lies within some set $E_\ell$) satisfies
\[
\EE[ \one_{E_\ell} ] = \bo(h_\ell^p),
\]
then for each $p\!>\!0$,
\[
\EE[\, |Y_\ell|\, \one_{E_\ell} ] = \bo(h_\ell^p).
\]
\end{lemma}
\begin{proof}
Immediate consequence of H\"{o}lder inequality which gives
\[
\EE[\, |Y_\ell| \one_{E_\ell} ] \leq 
\left( \EE[\, Y_\ell^2]  \right)^{1/2}
\left( \EE[\, \one_{E_\ell} ] \right)^{1/2}.
\]
\end{proof}

In the proofs in the main analysis, Lemma \ref{thm:ep1} will be 
used to establish the pre-conditions for Lemma \ref{thm:ep2}, 
from which it can be concluded, by choosing $p$ sufficiently large,
that the contribution of the extreme paths is negligible compared
to the paths that are not extreme.

\subsection{Analysis of the Milstein MLMC method}

\subsubsection{Brownian interpolation}

In all of the cases to be analysed, the discrete paths are simulated
using the Milstein method, with each level having twice as many timesteps 
as the previous level.  This gives a set of values at discrete times, 
$\hS_n \equiv \hS(t_n)$ where $t_n = n\, h$.  By approximating the drift 
and volatility as being constant within each timestep, we define the 
following Brownian interpolation based on equation (\ref{eq:bridge}), 
\begin{equation}
\hS(t) = \hS_n + \lambda(t)\,  (\hS_{n+1}\!-\!\hS_n) 
+ b_n\, \left(\rule{0in}{0.16in} W(t) - W_n - \lambda(t)\, (W_{n+1}\!-\!W_n) \right)
\label{eq:bridge2}
\end{equation}
where again $\lambda(t) \equiv (t - t_n)/h$. 
The advantage of this interpolation compared to the standard Kloeden-Platen 
interpolant is that we can use Lemmas \ref{thm:asian1} -- \ref{thm:barrier1}
in constructing the multilevel estimators.  The accuracy of the interpolant 
relative to the Kloeden-Platen interpolant is given by the following theorem:

\begin{theorem}\label{thm:KPBB}
If $\hS(t)$ is the interpolant defined by (\ref{eq:bridge2}) and $\hS_{KP}(t)$ is
the Kloeden-Platen interpolant defined by (\ref{eq:KP}) then for any positive
integer $m$
\begin{enumerate}
\item[i)]
\[\EE\left[\sup_{[0,T]} \left|\hS(t)-\hS_{KP}(t)\right|^m \right]= \bO( (h \added{|}\log h\added{|})^m),\]

\item[ii)]
\[\sup_{[0,T]}\EE\left[\, \left|\hS(t)-\hS_{KP}(t)\right|^m \right]= \bO(h^m),\]

\item[iii)]
\[
\EE\left[\left(\int_0^T(\hS(t)-\hS_{KP}(t))~\D t\right)^2\right]
=\bO(h^3).
\]
\end{enumerate}
\end{theorem}

\vspace{0.1in}

\begin{proof}
In each case we use the fact that $\EE[\max_n |b_n'b_n|^m]$ is finite 
for all positive integers $m$ 
due to Theorem \ref{thm:KP} and Assumption A2.  In addition, for 
$t\in [t_n,t_{n+1}]$, the difference between the two interpolants is
\[
\hS(t)-\hS_{KP}(t) =\fracs{1}{2}\, b_n' b_n\, Y(t),
\]
where
\begin{eqnarray*}
Y(t) &=& \lambda(t)\, (W_{n+1}\!-\!W_n)^2 - (W(t)\!-\!W_n)^2 \\
     &=& \lambda(t)\, (1\!-\!\lambda(t))\, (W_{n+1}\!-\!W_n)^2 - 
         \left(\rule{0in}{0.13in} W(t) - W_n - \lambda(t)\,(W_{n+1}\!-\!W_n) \right)^2\\
     & & -\ 2\, \lambda(t)\, (W_{n+1}\!-\!W_n)
         \left(\rule{0in}{0.13in} W(t) - W_n - \lambda(t)\,(W_{n+1}\!-\!W_n) \right)\, .
\end{eqnarray*}

\begin{enumerate}
\item[i)] Using H{\"o}lder's inequality, the assertion follows from
\[
\EE\left[\sup_{[0,T]} \left|\hS(t)-\hS_{KP}(t)\right|^m \right]
\leq 2^{-m}\sqrt{ \EE\left[\max_n |b_n'b_n|^{2m}\right] \ 
\EE\left[ \sup_{[0,T]} |Y(t)|^{2m} \right] },
\]
together with bounds on $\EE\left[ \sup_{[0,T]} |Y(t)|^{2m} \right]$ coming 
from Lemma \ref{thm:extreme2} and Corollary \ref{thm:extreme3}.

\item[ii)] 
By setting $W(t)\!-\!W_n = \sqrt{\lambda(t)\, h}\ Z_1$ and 
$W_{n+1}\!-\!W(t) = \sqrt{(1\!-\!\lambda(t))\, h}\ Z_2$, with $Z_1, Z_2$ independent 
standard Normal random variables, one can prove that 
$|Y| \leq h \max(Z_1^2, Z_2^2)$ and hence
$|Y|^m \leq h^m (Z_1^{2m} \!+\! Z_2^{2m})$.
The assertion then follows from
\[
\EE\left[\, \left|\hS(t)-\hS_{KP}(t)\right|^m\right] = 2^{-m} \EE[\,|b_n'b_n|^m]~\EE[\,|Y|^m],
\]
and standard results for moments of Normal random variables.

\item[iii)]  Defining
$\displaystyle
X_n:=\int_{t_n}^{t_{n+1}}Y(t)~\D t
$
we obtain
\[
\EE\left[\left(\int_0^T(\hS(t)-\hS_{KP}(t))~\D t\right)^2\right]
=\fracs{1}{4} \EE\left[ \left(
\sum_{n=0}^{N-1}b_n'b_nX_n\right)^2\right].
\]
For $n\!>\!m$, $\EE[b_m'b_mX_mb_n'b_nX_n]=0$ since $X_n$ is independent of 
$b_m'b_mX_mb_n'b_n$ and $\EE[X_n]=0$. 
In addition, the $X_n$ are iid random variables, and therefore
\[
\EE\left[\left(\int_0^T(\hS(t)-\hS_{KP}(t))~\D t\right)^2\right]
= \fracs{1}{4} \EE[X_0^2]\sum_{n=0}^{N-1}\EE[(b_n'b_n)^2].
\]
The proof is completed by noting that $\EE[X_0^2]=\bO(h^4)$ due to 
standard results for moments of Brownian increments.
\end{enumerate}
\end{proof}

\subsubsection{Estimator construction}

For each Brownian input, the multilevel estimator 
(\ref{eq:est_l}) requires the calculation of the payoff difference
$
\hP_\ell^f - \hP_{\ell-1}^c.
$
Here $\hP_\ell^f$ is a fine-path estimate using timestep $h_\ell \!=\! 2^{-\ell}T$, 
and $\hP_{\ell-1}^c$ is the corresponding coarse-path estimate using timestep 
$h \!=\! 2^{-(\ell-1)}T$.  As explained in \cite{giles08b}, to ensure that 
the identity (\ref{eq:identity}) is correctly respected, it is required
that \begin{equation}
\EE[\hP_{\ell-1}^f] = \EE[\hP_{\ell-1}^c].
\label{eq:equality}
\end{equation}
In the simplest case of a European option, this can be achieved 
very simply by defining $\hP_{\ell-1}^f$ and $\hP_{\ell-1}^c$ to be the same.
However, for the other applications the definition of $\hP_{\ell-1}^c$ 
involves information from the discrete simulation of $\hP_\ell^f$, 
which is not available in computing $\hP_{\ell-1}^f$.  This is done to 
reduce the variance of the estimator, but it must be shown that 
equality (\ref{eq:equality}) is satisfied.  This will be achieved 
in each case through a construction based on the Brownian interpolant.
In many cases this will involve evaluating the coarse path interpolant 
at the intermediate times $t_n$ for odd values of $n$, using the value 
for $W_n$ which was used for the fine path.

The analysis of the variance of the multilevel estimator will
often use the following decomposition of the difference
between the Brownian interpolants for the fine and coarse paths,
\begin{eqnarray}
 \hS^f(t) - \hS^c(t) &=& (\hS^f(t)\!-\!\hS^f_{KP}(t))
                    \ -\ (\hS^c(t)\!-\!\hS^c_{KP}(t)) \nonumber
 \\ && \!\!           +\ (\hS^f_{KP}(t)\!-\!S(t))
                    \ -\ (\hS^c_{KP}(t)\!-\!S(t))
\label{eq:decomposition}
\end{eqnarray}
with Theorem \ref{thm:KPBB} bounding the error in the first two terms, 
and Theorem \ref{thm:KP}  bounding the error in the last two terms.

\subsubsection{Lipschitz payoffs}
\label{sec:Lipschitz}

Many European options, such as simple put and call options, have a payoff 
that is a Lipschitz function of the value of the underlying asset at maturity,
$
P = f(S(T)).
$
Discrete Asian options have a payoff which is a Lipschitz function
of the value at maturity and the average of the underlying asset at a 
finite number of times $T_m$,
\[
\overline{S} = M^{-1} \sum_{m=1}^M S(T_m).
\]

Both of these are special cases of a more general class of Lipschitz 
payoffs in which the payoff is a Lipschitz function of the 
values of the underlying asset at a finite number of times $T_m$,
$
P = f(S(T_1), S(T_2), \ldots, S(T_M)),
$
with the Lipschitz bound
\[
\left| f(S^{(2)}_1, S^{(2)}_2, \ldots, S^{(2)}_M)
     - f(S^{(1)}_1, S^{(1)}_2, \ldots, S^{(1)}_M) \right|
\leq L \sum_{m=1}^M \left| S^{(2)}_m - S^{(1)}_m \right|,
\]
for some constant $L$.  In the numerical discretisation the fine 
and coarse path payoffs are both defined by
$
\hP = f(\hS(T_1), \hS(T_2), \ldots, \hS(T_M)),
$
with $\hS(t)$ given by the Brownian interpolation.  Note that
this will require the additional simulation of $W(T_m)$ if 
$T_m$ does not correspond to one of the existing timesteps.

We get the following result for the variance
of the multilevel estimator:

\begin{theorem}
For Lipschitz payoffs, $V_\ell = \bO(h_\ell^2)$.
\end{theorem}

\begin{proof}
From the Lipschitz bound and Jensen's inequality we obtain
\[
\VV[  \hP_\ell^{f} - \hP_{\ell-1}^{c}    ]\ \leq\ 
\EE[ (\hP_\ell^{f} - \hP_{\ell-1}^{c})^2 ]\ \leq\ 
L^2 M \sum_{m=1}^M \EE [ ( \hS^{f}(T_m) - \hS^{c}(T_m) )^2 ].
\]
The decomposition (\ref{eq:decomposition}) implies that
\begin{eqnarray*}
\lefteqn{ \EE[ (\hS^f(T_m) - \hS^c(T_m))^2 ]} \\
 &\!\leq\!& 4\left(
                            \EE[ (\hS^f(T_m)\!-\!\hS^f_{KP}(T_m))^2 ]
                         +\ \EE[ (\hS^c(T_m)\!-\!\hS^c_{KP}(T_m))^2 ]
\right. \\ && \left. \ \ +\ \EE[ (\hS^f_{KP}(T_m)\!-\!S(T_m))^2 ]
                          + \EE[ (\hS^c_{KP}(T_m)\!-\!S(T_m))^2 ]\, \right)
\end{eqnarray*}
and the proof is completed using Theorems \ref{thm:KP} and \ref{thm:KPBB}.
\end{proof}

\subsubsection{Asian options}

Continuously monitored Asian options have a payoff that is a uniform Lipschitz function 
of two arguments, the average over the time interval
\[
\overline{S} \equiv T^{-1} \int_0^T S(t) \ \D t,
\]
and the value at maturity, $S(T)$.

The numerical approximation follows the approach used in \cite{giles08b}, 
in which the fine and coarse path averages $\overline{\hS}$ are defined 
by integrating the interpolant (\ref{eq:bridge2}).  
Because of Lemma \ref{thm:asian1}, this gives
\[
\int_0^T \hS^f(t) \, \D t = \sum_{n=0}^{N-1}
 \fracs{1}{2} h_\ell \left( \hS^f_n + \hS^f_{n+1} \right)
+ b_n\, I^f_n
\]
where $I^f_n$ are independent $N(0,\fracs{1}{12}h_\ell^3)$ variables.
The payoff for the coarse path is defined similarly, but a straightforward 
calculation gives 
\begin{eqnarray*}
I^c_n &\equiv& \int_{t_n}^{t_{n+2}}  \left( W(t) - W_n - \frac{t-t_n}{2h_\ell}\, (W_{n+2}\!-\!W_n) \right)\, \D t
\\ &=& I^f_n + I^f_{n+1} - \fracs{1}{2} \, h_\ell\,( W_{n+2} - 2 W_{n+1} + W_n ).
\end{eqnarray*}
so $I^c_n$ is obtained from the Brownian path data used for the fine path.

\begin{theorem}
This approximation for continuous Asian payoffs has $V_\ell = \bO(h_\ell^2)$.
\label{thm:asian2}
\end{theorem}

\begin{proof}
Integrating (\ref{eq:decomposition}) gives
\begin{eqnarray*}
 \EE[ (\overline{\hS^f} - \overline{\hS^c})^2 ] &\leq& 4\left(
                 \  \EE[ (\overline{\hS^f}\!-\!\overline{\hS^f_{KP}})^2 ]
               \ +\ \EE[ (\overline{\hS^c}\!-\!\overline{\hS^c_{KP}})^2 ]
 \right. \\ && \left. \ \ 
               \ +\ \EE[ (\overline{\hS^f_{KP}}\!-\!\overline{S})^2 ]
               \ +\ \EE[ (\overline{\hS^c_{KP}}\!-\!\overline{S})^2 ]\ \right),
\end{eqnarray*}
and the proof is completed using the Lipschitz bound and Theorems \ref{thm:KP} and \ref{thm:KPBB}.
\end{proof}

\subsubsection{Lookback options}

In lookback options the payoff is a uniform Lipschitz function of the value 
of the underlying at maturity $S(T)$, and either the minimum or the maximum 
of the underlying over the time interval.  We will consider cases involving 
the minimum; the analysis for cases involving the maximum is very similar.

For the fine path simulation, we consider the conditional Brownian 
interpolation in the time interval $[t_n,t_{n+1}]$ defined by
\[
\hS^f(t) = \hS^f_n + \lambda(t)\,  (\hS^f_{n+1}\!-\!\hS^f_n) + 
b^f_n \left( W(t) - W_n - \lambda(t)\, (W_{n+1}\!-\!W_n) \right)
\]
where $\lambda(t) = (t-t_n)/h_\ell$ and $b^f_n\equiv b(\hS^f_n,t_n)$, and make use 
of Lemma \ref{thm:lookback1} to simulate the minimum on the time interval 
as
\begin{equation}
\hS^f_{n,min} = \fracs{1}{2} \left( \hS^f_n + \hS^f_{n+1} 
- \sqrt{ \left(\hS^f_{n+1} \!-\! \hS^f_n\right)^2 - 2\, (b^f_n)^2 \, h_\ell \log U_n }\ \right),
\label{eq:lookback1}
\end{equation}
where $U_n$ is a uniform random variable on the unit interval. Taking the minimum 
over all timesteps gives the global minimum which is used to compute the fine 
path value $\hP_\ell^f$.

For the coarse path value $\hP_{\ell-1}^c$, we do something slightly different.
Using the same conditional Brownian interpolation, for even $n$ we again use 
equation (\ref{eq:bridge2}) to define $\hS^c_{n+1}$.
The minimum value over the interval $[t_n,t_{n+2}]$ can then be taken to be the 
smaller of the minima for the two intervals $[t_n,t_{n+1}]$ and $[t_{n+1},t_{n+2}]$,
\begin{eqnarray}
\hS^c_{n,min}   &\!\!\!=\!\!& \fracs{1}{2}\! \left( \hS^c_n \!+\! \hS^c_{n+1} 
- \sqrt{ \left(\hS^c_{n+1} \!-\! \hS^c_n\right)^2 - 2\,(b^c_n)^2 \, h_\ell \log U_n }\ \right),
\nonumber \\
\hS^c_{n+1,min} &\!\!\!=\!\!& \!\fracs{1}{2}\! \left( \hS^c_{n+1} \!+\! \hS^c_{n+2} 
- \sqrt{\! \left(\hS^c_{n+2} \!-\! \hS^c_{n+1}\right)^2\!\! - 2\, (b^c_{n+1})^2 \, h_\ell \log U_{n+1} }\, \right).
\nonumber \\
\label{eq:lookback2}
\end{eqnarray}
Here $b^c_n =b^c_{n+1} \equiv b(\hS^c_n,t_n)$.  
Note the re-use of the same uniform random numbers $U_n$ and $U_{n+1}$ used to compute 
the fine path minimum.  Also, $\min(\hS^c_{n,min},\hS^c_{n+1,min})$ for level
$\ell$ has exactly the same distribution as $\hS^f_{n/2,min}$ for level $\ell\!-\!1$, 
since they are both based on the same approximate Brownian interpolation, and 
therefore equality (\ref{eq:equality}) is satisfied.

\begin{theorem}
\label{thm:lookback}
The multilevel approximation for a lookback option which is a uniform Lipschitz function of 
$S(T)$ and $\inf_{[0,T]} S(t)$ has $V_\ell = \bO(h_\ell^2 (\log h_\ell)^2)$.
\end{theorem}

\begin{proof}
If $\hS^f_{min}$ and $\hS^c_{min}$ are the computed minima 
for the fine and coarse paths, then
\[
\left| \hS^f_{min} - \hS^c_{min} \right| 
\leq \max_n \left| \hS^f_{n,min} - \hS^c_{n,min} \right|
\leq \max_n \left| \hS^f_{n} - \hS^c_{n} \right| +
     \max_n \left| \hD^f_{n} - \hD^c_{n} \right|,
\]
where 
\[
\hD^f_n = \fracs{1}{2} \sqrt{ \left(\hS^f_{n+1} \!-\! \hS^f_n\right)^2 - 2\,(b^f_n)^2 \, h_\ell \log U_n }
\]
and $\hD^c_n$ is defined similarly.   
If $\hD^f_n$ and $\hD^c_n$ are both zero, then $| \hD^f_{n} - \hD^c_{n} | =0$.
Otherwise, their sum is strictly positive and, using the inequality 
$\left| \rule{0in}{0.14in} |x| - |y| \right| \leq | x-y|$, 
straightforward manipulations give
\begin{eqnarray*}
\lefteqn{
\left| \hD^f_{n} - \hD^c_{n} \right| 
= \frac{\left| (\hD^f_{n})^2 - (\hD^c_{n})^2 \right|}{\hD^f_{n} +\hD^c_{n}}
} \\
&\leq & \frac{\left| (\hS^f_{n+1} - \hS^f_{n})^2 - (\hS^c_{n+1} - \hS^c_{n})^2 \right|}{4(\hD^f_{n} +\hD^c_{n})}
\ +\  \frac{|(b^f_{n})^2 -(b^c_{n})^2|\ h_\ell \, |\log U_n|}{2(\hD^f_{n} +\hD^c_{n})} \\
&\leq & \fracs{1}{2} \, \left| |\hS^f_{n+1} - \hS^f_{n}| -|\hS^c_{n+1} - \hS^c_{n}| \right|
\ +\  \fracs{1}{\sqrt{2}} \, \left| |b^f_{n}| -|b^c_{n}|\right|\ \sqrt{ h_\ell \, |\log U_n|} \\
&\leq & \fracs{1}{2} \, \left( \left| \hS^f_{n+1}  - \hS^c_{n+1} \right|
                             + \left| \hS^f_{n}    - \hS^c_{n}   \right| \right) 
\ +\  \fracs{1}{\sqrt{2}} \, |b^f_{n} -b^c_{n}|\ \sqrt{ h_\ell \, |\log U_n|},
\end{eqnarray*}
and hence
\[
\left( \hS^f_{min} \!-\! \hS^c_{min} \right)^2 
\leq 8 \max_n \left( \hS^f_{n} \!-\! \hS^c_{n} \right)^2 +
   h_\ell  \max_n  (b^f_{n} \!-\!b^c_{n})^2\ 
          \max_n |\log U_n|.
\]

When $n$ is even, assumption A1 gives
$
(b^f_{n} -b^c_{n})^2 \leq K_1^2 \left( \hS^f_n - \hS^c_n \right)^2,
$
while for odd $n$ we have
\begin{eqnarray*}
(b^f_{n} -b^c_{n})^2 
&=& \left( (b^f_{n} -b^f_{n-1})\ +\ (b^f_{n-1} -b^c_{n-1}) \right)^2\\
&\leq&  \replaced{2}{3} K_1^2 \left(  \hS^f_n - \hS^f_{n-1} \right)^2\added{+3K_3^2\left(  1+|\hS^f_{n-1}| \right)^2h}
\\&&
      + \replaced{2}{3} K_1^2 \left(  \hS^f_{n-1} - \hS^c_{n-1} \right)^2.
\end{eqnarray*}
Now,
\[
\hS^f_n - \hS^f_{n-1} = a_{n-1} h_\ell + b_{n-1} \Delta W_{n-1}
 + \fracs{1}{2} b'_{n-1}b_{n-1} ((\Delta W_{n-1})^2 - h_\ell).
\]
Asymptotically, the dominant term on the right is $b_{n-1} \Delta W_{n-1}$,
and it can be proved using the Jensen and H{\"o}lder inequalities, 
the boundedness of $\EE[\max_n b_n^4]$ and Lemma \ref{thm:extreme2} that
\[
\EE\left[ \max_n (\hS^f_n - \hS^f_{n-1})^2 \right] = \bO(h_\ell\, |\log h_\ell|),
\]
from which \added{together with Theorem \ref{thm:KP}} it follows that 
\[
\EE\left[ \max_n  (b^f_{n} -b^c_{n})^2   \right] = \bO(h_\ell\, |\log h_\ell|).
\]
From Lemma \ref{thm:extreme1},
\[
\EE\left[ \max_n |\log U_n| \right] = \bO(|\log h_\ell|),
\]
and hence, \added{using $\EE\left[\max\limits_n \left( \hS^f_{n} \!-\! \hS^c_{n} \right)^2\right]=\bO(h^2_\ell\, (\log h_\ell)^2)$,}
\[
\EE\left[ \left( \hS^f_{min} - \hS^c_{min} \right)^2 \right] = \bO(h^2_\ell\, (\log h_\ell)^2),
\]
and the final result then follows from the uniform Lipschitz property of the payoff
function and the bound
\[
\max_n \EE\left[ \left( \hS^f_n - \hS^c_n \right)^2 \right] = \bO(h^2_\ell).
\]

\end{proof}

\subsubsection{Extreme paths}

The analysis of the variance of the multilevel estimators for barrier and digital 
options will follow the extreme path approach used in \cite{ghm09}.  We prepare for 
this with the following lemma in which we use the notation
$
u \prec h^\alpha
$
when $u>0$ and there exists a constant $c>0$ such that
$
u < c \, h^\alpha,
$
for sufficiently small $h$.  Note that
\[
u_1 \prec h^{\alpha_1}, ~~~
u_2 \prec h^{\alpha_2} ~~~
\Longrightarrow ~~~
u_1 + u_2 \prec h^{\min(\alpha_1,\alpha_2)}, ~~~
u_1\, u_2 \prec h^{\alpha_1+\alpha_2}.
\]

\begin{lemma}
\label{thm:extreme_conditions}
For any $\gamma\!>\!0$, the probability that a Brownian path $W(t)$, its increments 
$\Delta W_n \equiv W((n\!+\!1)h)-W(nh)$, and the corresponding SDE solution $S(t)$ and
its fine $(h)$ and coarse $(2h)$ path approximations $\hS^f_n$ and $\hS^c_n$
satisfy any of the following extreme conditions
\begin{eqnarray*}
 \max_n\left( \max ( |S(nh)|,\ |\hS_n^f|,\ |\hS_n^c| ) \right) &>& h^{-\gamma}
\\  
\max_n \left( \max( |S(nh)\!-\!\hS_n^c|, \, 
                     |S(nh)\!-\!\hS_n^f|, \, 
                     |\hS_n^f\!-\!\hS_n^c| )\right) &>&  h^{1-\gamma}
\\ 
\max_n |\Delta W_n| &>& h^{1/2-\gamma}\\
\sup_{[0,T]} \left| \hS^f(t) - S(t)\right| &>& h^{1-\gamma}, \\
\sup_{[0,T]} \left| W(t) - \overline{W}(t)\right| &>& h^{1/2-\gamma},
\end{eqnarray*}
is $\bo(h^p)$ for all $p\!>\!0$.  Here $\overline{W}(t)$ is defined to be the 
piecewise linear interpolant of the discrete values $W_n$.

Furthermore, if none of these extreme conditions is satisfied, and
$\gamma<\fracs{1}{2}$, then
\begin{eqnarray}
\max_n |\hS^f_n - \hS^f_{n-1}| & \prec & h^{1/2-2\gamma}
\label{eq:ec1} \\
\max_n |b^f_n - b^f_{n-1}| & \prec & h^{1/2-2\gamma}
\label{eq:ec2} \\
\max_n \max(|b^f_n|,|b^c_n|) & \prec & h^{-\gamma}
\label{eq:ec3} \\
\max_n |b^f_n - b^c_n| & \prec & h^{1/2-2\gamma}
\label{eq:ec4}
\end{eqnarray}
where $b^c_n$ is defined to equal $b^c_{n-1}$ if $n$ is odd.

\end{lemma}

\begin{proof}
The probability of the first two extreme conditions is $\bo(h^p)$ for all 
$p\!>\!0$ due to Theorems \ref{thm:KP0}\added{,} \ref{thm:KP} \added{and \ref{thm:KPBB}(i)}  and Lemma \ref{thm:ep1}.
Since
\[
\PP\left(  \max_n |\Delta W_n| > h^{1/2-\gamma} \right)
\leq \sum_n \PP\left(|\Delta W_n| > h^{1/2-\gamma}\right),
\]
the probability of the third is $\bo(h^p)$ for all $p\!>\!0$ due to Lemma \ref{thm:ep1}.

The fourth extreme condition has a $\bo(h^p)$  probability because 
Theorems \ref{thm:KP} and \ref{thm:KPBB} together imply a uniform bound as 
$h\rightarrow 0$ for 
\[
\EE\left[ h^{-m+m\gamma/2} \sup_{[0,T]} \left| \hS^f(t) - S(t)\right|^m \right],
\]
for any $m>0$.  Similarly, the fifth is an extreme condition with $\bo(h^p)$  
probability because of Corollary \ref{thm:barrier2}.

If none of the extreme conditions is satisfied, then using Assumption A2 gives
\[
|\hS^f_{n+1} - \hS^f_n | < K_2 h ( 1 \!+\! h^{-\gamma})
                         + K_2 ( 1 \!+\! h^{-\gamma}) h^{1/2-\gamma} 
                         +  \fracs{1}{2} K_2( 1 \!+\! h^{-\gamma}) (h^{1-2\gamma} \!+\! h)
\]
and therefore (\ref{eq:ec1}) is satisfied provided $\gamma<\fracs{1}{2}$ so that 
$h^{1/2-2\gamma}$ is the dominant term in the above inequality.
  
(\ref{eq:ec2}) follows as a consequence because of Assumptions A1 and A3, 
and (\ref{eq:ec3}) is obtained similarly from Assumption A2 and the bound on 
$|\hS_n^f|$ and $|\hS_n^c|$.

When $n$ is even, the bound in (\ref{eq:ec4}) follows from Assumption A1 and the bound
on $|\hS_n^f\!-\!\hS_n^c|$, while for odd $n$ it requires the observation that
\[
|b^f_n - b^c_n|\  =  \ | b^f_n \!-\! b^c_{n-1} |
              \ \leq\ | b^f_n \!-\! b^f_{n-1} | + | b^f_{n-1} \!-\! b^c_{n-1} |
\]
and the bound then follows from (\ref{eq:ec2}) and the corresponding bound for $n\!-\!1$.

\end{proof}

\subsubsection{Barrier options}

The barrier option which is considered is a down-and-out option for 
which the payoff is a Lipschitz function of the value of the underlying
at maturity, provided the underlying has never dropped below a value $B$,
\[
P = f(S(T))\ {\bf 1}_{\tau\!>\!T},
\]
with the crossing time $\tau$ defined as 
$\displaystyle
\tau = \inf_{t>0}\left\{ S(t)<B \right\}.
$

One approach would be to follow the lookback approximation in computing the 
minimum of both the fine and coarse paths. However, the variance would be larger
in this case because the payoff is a discontinuous function of the minimum.
A better treatment, which is the one used in \cite{giles08b}, instead
computes for each timestep the probability that the minimum of the interpolant 
crosses the barrier, using the result from Lemma \ref{thm:barrier1}.  This gives 
the conditional expectation for the payoff, conditional on the discrete Brownian 
increments of the fine path.  For the fine path this gives
\[
\hP^f_\ell = f(\hS^f_N) \ \prod_{n=0}^{N-1} (1 - \hp^f_n),
\]
where
\[
\hp^f_n = \exp\left(\frac{ - 2\, (\hS^f_n \!-\! B)^+(\hS^f_{n+1} \!-\! B)^+ }{(b^f_n)^2\, h_\ell} \right).
\]
The payoff for the coarse path is similarly defined as
\[
\hP^c_{\ell\added{-1}} = f(\hS^c_N) \ \prod_{n=0}^{N-1} (1 - \hp^c_n),
\]
where
\begin{equation}
\hp^c_n = \exp\left(\frac{ - 2\, (\hS^c_n \!-\! B)^+(\hS^c_{n+1} \!-\! B)^+ }{(b^c_n)^2\, h_\ell} \right),
\label{eq:hpc}
\end{equation}
and for odd values of $n$, $\hS^c_n$ is defined by the usual interpolant and 
$b^c_n\equiv b^c_{n-1}$.

Equality (\ref{eq:equality}) is satisfied in this case because
\[
\PP\left( \inf_{[t_n,t_{n+2}]} \hS^c(t) \!>\! B \ | \ \hS^c_n, \hS^c_{n+2}\right)
= \EE\left[
\PP\left( \inf_{[t_n,t_{n+2}]} \hS^c(t) \!>\! B \ | \ \hS^c_n, \hS^c_{n+1}, \hS^c_{n+2}\right)
\right]
\]
where the expectation on the r.h.s.~is taken with respect to the distribution 
of the interpolated value $\hS^c_{n+1}$, conditional on $\hS^c_n, \hS^c_{n+2}$.  

\begin{theorem}
Provided
$\displaystyle
b_{min} \equiv \inf_{[0,T]} |b(B,t)| > 0,
$
and 
$\displaystyle
\inf_{[0,T]} S(t)
$ 
has a bounded density in the neighbourhood of $B$, then
the multilevel estimator for a down-and-out barrier option has variance 
$V_\ell = \bo(h_\ell^{3/2-\delta})$ for any $\delta\!>\!0$.
\label{thm:barrier}
\end{theorem}

\begin{proof}
The proof involves dividing the paths into the following three subsets:\\[0.05in]
(i) extreme paths;\\[0.05in]
(ii) paths which are not extreme and for which
$
|S_{min} \!-\! B\, | > h_\ell^{1/2 - 4 \gamma}
$
for $0\!<\!\gamma\!<\!\fracs{1}{8}$;
\\[0.05in]
(iii) the rest.\\[0.05in]

Following the extreme path approach used in \cite{ghm09}, we start with
\begin{eqnarray*}
\lefteqn{
\VV[  \hP_\ell^{f} \!-\! \hP_{\ell-1}^{c}    ]
}\\
& \leq & \EE[ (\hP_\ell^{f} \!-\! \hP_{\ell-1}^{c})^2 ] \\
& = & \EE[ (\hP_\ell^{f} \!-\! \hP_{\ell-1}^{c})^2 {\bf 1}_{(i)} ]
\ +\  \EE[ (\hP_\ell^{f} \!-\! \hP_{\ell-1}^{c})^2 {\bf 1}_{(ii)} ]
\ +\  \EE[ (\hP_\ell^{f} \!-\! \hP_{\ell-1}^{c})^2 {\bf 1}_{(iii)} ]
\end{eqnarray*}
where the indicator functions have unit value for paths within the respective
subsets.  Each of these is considered in turn, and their contributions to
$\EE[(\hP^f_\ell \!-\! \hP^c_{\ell-1})^2]$ are bounded.

\vspace{0.1in}
\noindent
(i) Paths are defined to be extreme if they satisfy any of the conditions of
Lemma \ref{thm:extreme_conditions} for $0\!<\!\gamma\!<\!\fracs{1}{8}$.
The Lipschitz bound for the payoff together with the bounds in Theorem
\ref{thm:KP} imply a uniform bound for $\EE[\, (\hP_\ell^{f})^4]$ and 
$\EE[\, (\hP_{\ell-1}^{c})^4\,]$ and therefore also for 
$\EE[\, (\hP_\ell^{f} \!-\! \hP_{\ell-1}^{c})^4\,]$.
Hence, by \replaced{Theorem}{Lemma} \ref{thm:ep2}, 
$\EE[ (\hP_\ell^{f} \!-\! \hP_{\ell-1}^{c})^2 {\bf 1}_{(i)}]$ 
is $\bo(h_\ell^p)$ for all $p\!>\!0$.

\vspace{0.1in}
\noindent
(ii) Suppose that $S(t)$ attains its minimum at time $\tau \in [t_n, t_{n+1}]$.

First we consider the case $S_{min} < B - h_\ell^{1/2-4\gamma}$. Starting with
\[
| \hS^f_n - S_{min} | \leq | \hS^f_n - \hS^f(\tau) | + |\hS^f(\tau) - S(\tau) |,
\]
and noting that
\[
\hS^f(\tau) - \hS^f_n = \frac{\tau - t_n}{h} \left( \hS^f_{n+1} - \hS^f_n \right)
+ b_n^f (W(\tau) - \overline{W}(\tau)),
\]
we can conclude that $| \hS^f_n - S_{min} | \prec h_\ell^{1/2 - 2\gamma}$.
Hence, for sufficiently small $h_\ell$,
 $| \hS^f_n - S_{min} | < h_\ell^{1/2-4\gamma}$ and so $\hS^f_n$ is guaranteed to be 
less than $B$.
In addition, $\hS^f_n \!-\! \hS^c_n < h_\ell^{1-\gamma}$ and so, for sufficiently 
small $h_\ell$, $\hS^c_n$ is also guaranteed to be less than $B$ and hence
$\hP_\ell^f \!-\! \hP_{\ell-1}^c = 0$.

In the alternate case $S_{min} > B + h_\ell^{1/2-4\gamma}$, then
\[
\min_n \min(\hS^f_n, \hS^c_n) > B + h_\ell^{1/2-4\gamma} - h_\ell^{1-\gamma}
\]
and since $h_\ell^{1-\gamma} \prec h_\ell^{1/2-4\gamma}$ it follows that
$\prod_n (1\!-\!\hp^f_n)$ and $\prod_n (1\!-\!\hp^c_n)$ 
are both equal to $1 - \bo(h_\ell^p)$ for all $p\!>\!0$, and so
$
|\hP_\ell^f \!-\! \hP_{\ell-1}^c| \prec h_\ell^{1-\gamma}
$
due to the Lipschitz condition and the bound on $\hS_N^f \!-\! \hS_N^c$.
Hence, the contribution to $\EE[(\hP^f_\ell \!-\! \hP^c_{\ell-1})^2]$ is  
at most $\bO(h_\ell^{2-2\gamma})$.

\vspace{0.1in}
\noindent
(iii) Our first step is to note that if any one of 
$\hS^f_n,\hS^f_{n+1}, \hS^c_n,\hS^c_{n+1}$ is greater than $B+h_\ell^{1/2-3\gamma}$,
then the others will be greater than $B+\fracs{1}{2} h_\ell^{1/2-3\gamma}$,
when $h_\ell$ is sufficiently small, since $|\hS^f_n-\hS^f_{n+1}|\prec h_\ell^{1/2-2\gamma}$
and $\max(|\hS^f_n-\hS^c_n|,|\hS^f_{n+1}-\hS^c_{n+1}| ) \leq h_\ell^{1-\gamma}$.
In this case, $\hp^f_n$ and $\hp^c_n$ will both be $\bo(h_\ell^p)$, and so 
\[
\prod_n (1\!-\!\hp^f_n) = \prod_{n\in R} (1\!-\!\hp^f_n) + \bo(h_\ell^p),
\]
and
\[
\prod_n (1\!-\!\hp^c_n) = \prod_{n\in R} (1\!-\!\hp^c_n) + \bo(h_\ell^p),
\]
where $R$ is the set of indices $n$ for which none of 
$\hS^f_n,\hS^f_{n+1}, \hS^c_n,\hS^c_{n+1}$ 
is greater than $B+h_\ell^{1/2-3\gamma}$.

Assume $n\in R$.  We have
$
\hS^f_n - \hS^c_n < h_\ell^{1-\gamma}
$
and
$
\hS^f_{n+1} - \hS^c_{n+1} < h_\ell^{1-\gamma}
$
due to the definition of extreme paths, and
$
b^f_n - b^c_n \prec h_\ell^{1/2-2\gamma},
$
due to Lemma \ref{thm:extreme_conditions}.
If we now define
\begin{eqnarray*}
X^f_n &\equiv& \frac{2\, (\hS^f_n \!-\! B)^+(\hS^f_{n+1} \!-\! B)^+ }{(b^f_n)^2\, h_\ell},\\
X^c_n &\equiv& \frac{2\, (\hS^c_n \!-\! B)^+(\hS^c_{n+1} \!-\! B)^+ }{(b^c_n)^2\, h_\ell},
\end{eqnarray*}
then when $X^f_n$ and $X^c_n$ are both strictly positive it follows, 
through the continuity of $b(S,t)$ and for sufficiently small $h_\ell$, that 
$
\min (|b^c_n|, |b^f_n|) > \fracs{1}{2} \, b_{min},
$
and hence through repeated use of the following identity,
\begin{equation}
f_1\, g_1 - f_2\, g_2 = \fracs{1}{2} (f_1\!-\!f_2)(g_1\!+\!g_2)
                      + \fracs{1}{2} (f_1\!+\!f_2)(g_1\!-\!g_2),
\label{eq:diff}
\end{equation}
and the fact that $n\in R$ to bound terms such as $\hS^f_n\!-\!B$, we obtain\footnote{\added{
With $u_{f/c}=(\hS^{f/c}_n \!-\! B)^+$, $v_{f/c}=(\hS^{f/c}_{n+1} \!-\! B)^+$ and some suitable $C$ we obtain
$\left|  X^f_n  - X^c_n  \right|
= \frac{2}{h_\ell}\left| \frac{u_fv_f-u_cv_c}{(b^f_n)^2}+u_cv_c\left(\frac{(b^c_n-b^f_n)(b^c_n+b^f_n)}{(b^f_n)^2(b^c_n)^2}\right) \right|
\leq Ch_\ell^{-1}\left|u_fv_f-u_cv_c\right|+Ch_\ell^{-1}u_cv_c\left|b^c_n-b^f_n\right|
\prec h_\ell^{1/2 - 4 \gamma}+h_\ell^{1/2 - 8 \gamma}\prec h_\ell^{1/2 - 8 \gamma}.
$}
} 
$
\left|  X^f_n  - X^c_n  \right| \prec h_\ell^{1/2 - \replaced{4}{8} \gamma},
$
and hence
$
\left|  X^f_n  - X^c_n  \right| < h_\ell^{1/2 - \replaced{5}{9} \gamma},
$
for sufficiently small $h_\ell$.
The same bound can also be achieved in the other cases in which at least one 
of $X^f_n$ and $X^c_n$ is equal to zero.
If we define
\[
\Delta_\ell\ \equiv\ 1 - \exp(- h_\ell^{1/2 -  \replaced{5}{9} \gamma}),
\]
then we obtain
\begin{eqnarray*}
1 - \hp^c_n &=& (1 - \hp^f_n) + (\hp^f_n - \hp^c_n) \\
            &=& (1 - \hp^f_n) + \hp^f_n \left( 1 - \exp( X^f_n  \!-\! X^c_n) \right)\\ 
            &\leq & (1 - \hp^f_n) + \hp^f_n \Delta_\ell.
\end{eqnarray*}

Since 
$\displaystyle
g(\Delta)\ \equiv\ \prod_{n\in R} \left(\rule{0in}{0.15in} (1\!-\!p^f_n) + p^f_n \,\Delta \right) - \prod_{n\in R} (1\!-\!p^f_n)  -  \Delta
$
is convex, $g(0)\!=\!0$ and $g(1)\!=\!- \prod_{n\in R} (1\!-\!p^f_n) \leq 0$, we conclude
that 
$g(\Delta)\!\leq\!0$, $\forall \Delta \in [0,1]$.
Hence, 
\[
\prod_{n\in R} (1 \!-\! \hp^c_n)\ \leq\ \prod_{n\in R} (1 \!-\! \hp^f_n)\ +\ \Delta_\ell.
\]
Similarly, 
$
1 - \hp^f_n \leq  (1 - \hp^c_n) + \hp^c_n \Delta_\ell,
$
which leads to
\[
\prod_{n\in R} (1 \!-\! \hp^f_n)\ \leq\ \prod_{n\in R} (1 \!-\! \hp^c_n)\ +\ \Delta_\ell,
\]
and therefore
\[
\left| \prod_{n\in R} (1 \!-\! \hp^f_n) - \prod_{n\in R} (1 \!-\! \hp^c_n) \right| 
\leq \Delta_\ell.
\]

Returning to the original products over all $n$, 
\[
\left| \prod_n (1 \!-\! \hp^f_n) - \prod_n (1 \!-\! \hp^c_n) \right| 
\prec h_\ell^{1/2 -  \replaced{5}{9} \gamma}.
\]
This gives us $\hP^f_\ell \!-\!\hP^c_{\ell-1} \prec h_\ell^{1/2 - \replaced{6}{10} \gamma}$,
because of the bound on $f(\hS^f_N)$ and $f(\hS^c_N)$, and so the contribution
from set (iii) to $\EE[(\hP^f_\ell \!-\! \hP^c_{\ell-1})^2]$ is at most
$\bO(h_\ell^{3/2-\replaced{16}{24}\gamma})$.

\vspace{0.1in}

The proof is finally completed by choosing $\gamma < \min(\fracs{1}{8},\delta/\replaced{16}{24})$.

\end{proof}

\subsubsection{Digital options}

A digital option has a payoff which is a discontinuous function of the 
value of the underlying asset at maturity, the simplest example being
\[
P = {\bf 1}_{S(T)>K},
\]
which has a unit payoff iff $S(T)$ is greater than the strike $K$.

The difficulty with the digital option is that the approach used in section 
\ref{sec:Lipschitz} will lead to an $\bO(h_\ell)$ fraction of the paths having 
coarse and fine path approximations to $S(T)$ on either side of the strike, 
producing $\hP^f_\ell - \hP^c_{\ell-1} = \pm 1$, resulting in $V_\ell = \bO(h_\ell)$.
To improve the variance to $\bO(h_\ell^{3/2-\delta})$ for all $\delta \!>\!0$
we follow the approach which was tested numerically in \cite{giles08b},
using the technique of conditional expectation (see section 7.2.3 in 
\cite{glasserman04}).

If $\hS^f_{N-1}$ denotes the value of the fine path approximation 
one timestep before maturity, then if we approximate the motion thereafter 
as a simple Brownian motion with constant drift 
$a^f_{N-1}\!\equiv\!a(\hS^f_{N-1},T\!-\!h_\ell)$ and volatility 
$b^f_{N-1}\!\equiv\!b(\hS^f_{N-1},T\!-\!h_\ell)$, the conditional 
expectation for the payoff is the probability that $\hS^f_N \!>\! K$ after 
one further timestep, which is
\begin{equation}
\hP_\ell^f = \Phi \left( \frac{\hS^f_{N-1} \!+\! a^f_{N-1} h_\ell - K}{|b^f_{N-1}|\, \sqrt{h_\ell}}\right),
\label{eq:digital1}
\end{equation}
where $\Phi$ is the cumulative Normal distribution.

For the coarse path, we note that given the Brownian increment $\Delta W_{N-2}$ for the 
first half of the last coarse timestep (which comes from the fine path simulation), 
the probability that $\hS^c_N \!>\! K$ is
\begin{equation}
\hP_{\ell-1}^c = \Phi \left( \frac{\hS^c_{N-2} \!+\! 2 a^c_{N-2} h_\ell \!+\! b^c_{N-2} \Delta W_{N-2} - K}
{|b^c_{N-2}|\, \sqrt{h_\ell}}\right).
\label{eq:digital2}
\end{equation}
The conditional expectation of (\ref{eq:digital2}) is equal to the conditional expectation of 
$\hP_{\ell-1}^f$ defined by (\ref{eq:digital1}) on level $\ell\!-\!1$, and so equality (\ref{eq:equality}) 
is satisfied.

A bound on the variance of the multilevel estimator is given by the following result:

\begin{theorem}
Provided $b(K,T) \neq 0$, and $S(t)$ has a bounded density in the neighbourhood of $K$, 
then the multilevel estimator for a digital option has variance 
$V_\ell = \bo(h_\ell^{3/2-\delta})$ for any $\delta\!>\!0$.
\end{theorem}

\begin{proof}
As in the proof of Theorem \ref{thm:barrier}, we split the paths into three subsets:\\[0.05in]
(i) extreme paths;\\[0.05in]
(ii) paths which are not extreme and for which
$
|S_N\!-\! K\, | > h_\ell^{1/2 - \replaced{4}{3} \gamma};
$\\[0.05in]
(iii) the rest\\[0.05in]
and we analyse the contributions to $\EE[ (\hP_\ell^{f} \!-\! \hP_{\ell-1}^{c})^2]$ 
from all three subsets.

\vspace{0.1in}

(i) Paths are defined to be extreme if they satisfy any of the conditions of
Lemma \ref{thm:extreme_conditions} for $0<\gamma<\fracs{1}{4}$.
$\EE[ (\hP^f)^4]$ and $\EE[ (\hP^c)^4]$ are both finite, and hence the contribution of 
the extreme paths is $\bo(h_\ell^p)$, for all $p>0$.

\vspace{0.1in}

(ii)
If we define $\hS^f_N$ and $\hS^c_N$ to be the values which we would have obtained 
from the fine and coarse path simulations after the final timestep, then
\begin{eqnarray*}
\lefteqn{
\frac{\hS^f_{N-1} \!+\! a^f_{N-1} h_\ell - K}{|b^f_{N-1}|\, \sqrt{h_\ell}}
~~~=~~~ \frac{\hS^f_N - K}{|b^f_{N-1}|\, \sqrt{h_\ell}} 
}
\\&&
-\ \frac{b^f_{N-1}}{|b^f_{N-1}|\, \sqrt{h_\ell}} \left(
\Delta W_{N-1} + \fracs{1}{2} (b')^f_{N-1}
\left(\rule{0in}{0.14in} (\Delta W_{N-1})^2 \!-\! h_\ell \right)\right),
\end{eqnarray*}
and similarly
\begin{eqnarray*}
\lefteqn{
\frac{\hS^c_{N-2} \!+\! 2 a^c_{N-2} h_\ell \!+\!b^c_{N-2}\Delta W_{N-2} - K}{|b^c_{N-2}|\, \sqrt{h_\ell}}
~~~=~~~ \frac{\hS^c_N - K}{|b^c_{N-2}|\,\sqrt{h_\ell}} 
}
\\&&
-\ \frac{b^c_{N-2}}{|b^c_{N-2}|\, \sqrt{h_\ell}} \left(\rule{0in}{0.15in}
\Delta W_{N-1} + \fracs{1}{2}  (b')^c_{N-2}
\left(\rule{0in}{0.14in} (\Delta W_{N-2}\!+\!\Delta W_{N-1})^2 \!-\! 2 h_\ell\right) \right).
\end{eqnarray*}

Since the paths are not extreme, 
$|\Delta W_n| \leq h_\ell^{1/2-\gamma}$ and
$|S(T)\!-\!\hS^f_N| \leq h_\ell^{1-\gamma}$, 
and due to Lemma \ref{thm:extreme_conditions} 
$|b^f_{N-1}|\prec h_\ell^{-\gamma}$.
Consequently, if $S(T) > K + h_\ell^{1/2-3\gamma}$ then for sufficiently small $h_\ell$ 
it follows that
\[
\frac{\hS^f_{N-1} \!+\! a^f_{N-1} h_\ell - K}{|b^f_{N-1}| \, \sqrt{h_\ell}} > C\, h_\ell^{-2\gamma},
\]
for some suitably chosen constant $C$.
A similar result follows for the corresponding coarse path, and hence for these paths
$
\hP_\ell^{f} - \hP_{\ell-1}^{c} = \bo(h_\ell^p), 
$
for all $p\!>\!0$.  A similar argument applies to the other paths for which 
$S(T)<K-h_\ell^{1/2-3\gamma}$, and hence
$
\EE[ (\hP_\ell^{f} - \hP_{\ell-1}^{c})^2 {\bf 1}_{(ii)}]
$
is $\bo(h_\ell^p)$ for all $p\!>\!0$ and so this contribution is also negligible.

\vspace{0.1in}

(iii)
This subset consists of non-extreme paths for which $|S(T)-K| \leq h_\ell^{1/2-3\gamma}$.
Since
\[
b^f_{N-1} - b(K,T) = (b^f_{N-1}\!-\!b^f_N) + (b^f_N \!-\! b(K,T)),
\]
using Assumption A1 and (\ref{eq:ec2}) with $\gamma \!<\! \fracs{1}{\replaced{4}{6}}$ 
we can conclude that for sufficiently small $h_\ell$,  
$
|b^f_{N-1} - b(K,T)| < \fracs{1}{2}\, |b(K,T)|
$
and in particular $b^f_{N-1}$ is non-zero and of the same sign as $b(K,T)$.
The same also applies to $b^c_{N-2}$ and hence, exploiting the Lipschitz property
$|\Phi(x_1)-\Phi(x_2)| \leq |x_1\!-\!x_2|$,
\begin{eqnarray*}
\lefteqn{ \left| \hP_\ell^{f} - \hP_{\ell-1}^{c} \right| }\\
&\leq& \left| \frac{\hS^f_{N-1} \!+\! a^f_{N-1} h_\ell - K}{|b^f_{N-1}|\, \sqrt{h_\ell}}
\ -\ \frac{\hS^c_{N-2} \!+\! 2 a^c_{N-2} h_\ell \!+\!b^c_{N-2}\Delta W_{N-2} - K}{|b^c_{N-2}|\, \sqrt{h_\ell}}
\right|
\\ &\leq&
\left| \frac{\hS^f_N - K}{|b^f_{N-1}|\, \sqrt{h_\ell}}  - \frac{\hS^c_N - K}{|b^c_{N-2}|\, \sqrt{h_\ell}}  \right|
\\ &&
+ \ \fracs{1}{2}  K_1 h_\ell^{-1/2} \left\{ (\Delta W_{N-1})^2
                                     + (\Delta W_{N-2}+\Delta W_{N-1})^2 + 3\, h_\ell \right\}.
\end{eqnarray*}

Using the identity (\ref{eq:diff}) we obtain
\begin{eqnarray*}
 \frac{\hS^f_N - K}{|b^f_{N-1}|\, \sqrt{h_\ell}}  - \frac{\hS^c_N - K}{|b^c_{N-2}|\, \sqrt{h_\ell}} 
&=& \ \ \frac{1}{2 \sqrt{h_\ell}}\, 
(\hS^f_N -\hS^c_N ) \left(\frac{1}{|b^f_{N-1}|} + \frac{1}{|b^c_{N-2}|}\right)
\\&+& \frac{1}{2 \sqrt{h_\ell}}\, 
(\hS^f_N \!+\! \hS^c_N \!-\! 2K) \left(\frac{|b^c_{N-2}| \!-\! |b^f_{N-1}|}{|b^f_{N-1}| \ |b^c_{N-2}|}\right).
\end{eqnarray*}
Using the bounds provided by Lemma \ref{thm:extreme_conditions}, it follows that
\[
\frac{\hS^f_N - K}{|b^f_{N-1}|\, \sqrt{h_\ell}}  - \frac{\hS^c_N - K}{|b^c_{N-2}|\, \sqrt{h_\ell}} 
= \bO(h_\ell^{1/2 - 5 \gamma}),
\]
and hence
$
\hP^f_\ell - \hP^c_{\ell-1} = \bO(h_\ell^{1/2 - 5 \gamma}).
$
Since $\EE[{\bf 1}_{(iii)}] = \bO(h_\ell^{1/2 - 3\gamma})$ due to the bounded
probability density for $S(T)$, it follows that
$
\EE[(\hP^f_\ell \!-\! \hP^c_{\ell-1})^2{\bf 1}_{(iii)}] = \bO(h_\ell^{3/2 - 13 \gamma}).
$
Choosing $\gamma < \min(\fracs{1}{\replaced{4}{6}},\delta/13)$ completes the proof.

\end{proof}

\section{Conclusions and future work}

In this paper we have proved that when using the Milstein discretisation 
for a scalar SDE the variance of the multilevel estimator is $\bO(h_\ell^2)$
for Lipschitz and Asian options, $\bO(h_\ell^2 (\log h_\ell)^2)$ for lookback options,
and  $\bo(h_\ell^{3/2-\delta})$ for barrier and digital options, for any $\delta>0$.

The MLMC theorem \cite{giles08} also requires knowledge of the order of weak 
convergence.  Theorems \ref{thm:KP} and \ref{thm:KPBB} together give $\bO(h)$ 
weak convergence for the Lipschitz and Asian options, and $\bO(h\added{|} \log h\added{|})$ 
convergence for the lookback option.  For the digital and barrier options, 
the analysis of the multilevel convergence can be modified to instead consider 
$\EE[ \hP_\ell \!-\! P ]$, and hence it can be proved that the weak order of 
convergence is $\bo(h^{1-\delta})$ for any $\delta>0$.  From this, in all cases
considered in this paper it can be concluded that the computational cost to 
achieve a RMS accuracy of $\eps$ is $\bO(\eps^{-2})$.

The challenge for the future is to construct and analyse effective MLMC 
estimators for multi-dimensional SDEs which do not satisfy the commutativity
condition, and therefore would require the simulation of L{\'e}vy areas to
achieve first order strong convergence.  Giles and Szpruch \cite{gs14} have 
proved that if one ignores the L{\'e}vy area terms, it is nevertheless
possible to construct an efficient antithetic MLMC estimator, despite the 
strong convergence being of the same order as the Euler-Maruyama discretisation.
They prove $V_\ell = O(h_\ell^2)$ for a European payoff which is twice 
differentiable, and $V_\ell = O(h_\ell^{3/2})$ for a payoff such as a put or 
call function which is continuous but not everywhere twice-differentiable.
However, it is still an open problem to construct and analyse good estimators 
for lookback, barrier and digital options for general systems of SDEs.


\providecommand{\href}[2]{#2}
\providecommand{\arxiv}[1]{\href{http://arxiv.org/abs/#1}{arXiv:#1}}
\providecommand{\url}[1]{\texttt{#1}}
\providecommand{\urlprefix}{URL }

\end{document}